\documentclass[journal]{IEEEtran}

\usepackage[T1]{fontenc}
\usepackage{amsmath,amssymb,amsthm,mathtools}
\usepackage{booktabs,longtable,array}
\usepackage{placeins}
\usepackage{tikz}
\usetikzlibrary{arrows.meta}
\usepackage{microtype}
\usepackage[hidelinks]{hyperref}

\allowdisplaybreaks[3]
\newtheorem{theorem}{Theorem}
\newtheorem{lemma}{Lemma}
\newtheorem{proposition}{Proposition}
\newtheorem{corollary}{Corollary}
\theoremstyle{definition}
\newtheorem{definition}{Definition}
\newtheorem{example}{Example}
\theoremstyle{remark}
\newtheorem{remark}{Remark}

\author{Yi~Wang and Linglong~Dai, {\textit{Fellow, IEEE}}%
\thanks{This work was supported in part by the National Science Fund for Distinguished Young Scholars under Grant 62325106, in part by the National Science and Technology Major Projects of China under Grant 2025ZD1301800, and in part by the National Key Research and Development Program of China under Grant 2023YFB3811503.}%
\thanks{The authors are with the Department of Electronic Engineering, Tsinghua University, and the State Key Laboratory of Space Network and Communications, Tsinghua University, Beijing 100084, China (e-mails: yiwang24@mails.tsinghua.edu.cn, daill@tsinghua.edu.cn).}}

\title{From Source Reconstruction to Predictive State Preservation: An Information-Theoretic Framework for AI-Native Communication}
\hypersetup{
  pdftitle={From Source Reconstruction to Predictive State Preservation: An Information-Theoretic Framework for AI-Native Communication},
  pdfauthor={Yi Wang and Linglong Dai}
}

\begin{document}

\maketitle

\begin{abstract}
AI-native communication increasingly aims to support prediction rather than reproduce every detail of the source. This shift raises a basic question left implicit by conventional source coding: what should be preserved when the terminal goal is prediction? We take the source-induced predictive state as the fidelity object. It is the distribution of the specified future conditioned on the source observation and shared context. We show that this state is sufficient and minimal for exact predictive preservation. The terminal prediction loss then defines communication distortion as lost predictive performance rather than source reconstruction error. Under logarithmic loss, this distortion equals the conditional mutual information lost through communication. Using the receiver-side predictive state as a Bayes reference, we separate AI-receiver error into predictive value lost in communication, receiver-available value unusable by the model family, and family capability not realized by the deployed model. The same predictive state also suffices for matched compression. For finite-alphabet memoryless sources, access to the raw source gives no rate-distortion advantage over coding the state directly. Applying the same target-conditioned construction to sequential prediction reveals a dynamic boundary. A state induced by a fixed horizon is minimal for that horizon but may not support recursive updating as the target window shifts. Taking the entire future as the target yields a minimal full-future state that updates recursively and admits a Markov representation. Together, these results shift AI communication from source reconstruction to predictive-state preservation.
\end{abstract}


\section{Introduction}

AI-native communication increasingly serves terminals that use received information to predict future outcomes. Classical source-coding theory characterizes how efficiently a prescribed fidelity object can be preserved, but it does not determine what that object should be. Conventional source reconstruction takes the source realization itself as the fidelity object. For a predictive terminal, however, distinct source realizations should be treated as equivalent whenever they induce the same predictive distribution for the specified future. The fidelity object for predictive communication must therefore be identified before distortion and coding are defined.

We identify this object as the source-induced predictive state. Let $X$ denote the source observation, $C$ the shared context, and $F$ the specified prediction target. Define
\[
\boxed{
S_X
\triangleq
P_{F\mid X,C}.
}
\]
The state $S_X$ records the full predictive distribution of $F$ available from the source observation and shared context.

If two source realizations $x_1$ and $x_2$ satisfy
\[
P_{F\mid X,C}(\,\cdot\mid x_1,c)
=
P_{F\mid X,C}(\,\cdot\mid x_2,c),
\]
then they are indistinguishable for predicting $F$, even if a reconstruction metric distinguishes them.

Thus $S_X$ realizes the target-conditioned quotient induced by the prediction objective. It is not a learned embedding selected by the system designer, nor must it be transmitted literally as a probability vector. Rather, it specifies the invariant that every representation preserving the full predictive distribution must retain. This recasts prediction-oriented communication from source reconstruction to predictive-state preservation. The remainder of the paper develops the fidelity, receiver, coding, and sequential consequences of this change in fidelity object.

\subsection{Relation to Prior Work}

Classical source coding and indirect source coding assume that the source object or remote target has been specified before the coding problem is posed \cite{shannon1959fidelity,gray1972conditional,wolf_ziv1970remote,witsenhausen1980indirect}. Recent semantic-communication theories move this choice earlier by constructing task-dependent alphabets or equivalence classes and by defining fidelity between semantic distributions \cite{liu_shao_zhang_poor2022semantic,niu_zhang2024semantic,zhao_et_al2025semanticrd,nixon2026bounded}. These works move semantic communication from coding a prescribed source alphabet toward constructing a task-dependent fidelity object.

The use of predictive distributions as sufficient representations is well established in the information-bottleneck and statistical-relevance literature \cite{tishby1999ib,shalizi_crutchfield2002transduction}. Recent work further shows that reducing a source to a sufficient statistic can preserve both information-bottleneck tradeoffs and finite-block log-loss coding performance \cite{armstrong2026sufficiency}. Decision-theoretic work also shows that a coarser representation may suffice when the terminal needs only a particular decision rather than the full predictive distribution \cite{sevetlidis2026bayes}. These results distinguish full predictive-distribution preservation from the weaker preservation of a task-specific Bayes action.

Proper scoring rules evaluate probabilistic predictions and quantify the excess risk incurred by using one predictive distribution in place of another \cite{savage1971elicitation,gneiting_raftery2007scoring}. Under logarithmic loss, this excess risk becomes relative entropy, and the predictive value of side information is characterized by mutual information \cite{courtade_weissman2014logloss,jiao2015logloss}. These results connect probabilistic prediction error to divergence and information loss, providing the loss-based ingredients needed for predictive fidelity.

Predictive $\mathcal V$-information and the Decodable Information Bottleneck show that information present in a representation may still be unusable by a restricted predictor family \cite{xu_et_al2020usable,dubois_et_al2020dib}. Related work has also separated information loss from predictor error, shown that a fixed AI decoder may fail to exploit information already available at the receiver, and studied rate-distortion limits under model constraints \cite{charpentier_fernandes_machado2026scores,billa2026modality,enttsel_corlay2026modelaware}. These results motivate distinguishing receiver-available information from model-usable predictive performance; we retain this distinction through a restricted Bayes-risk gap.

For sequential prediction, causal-state theory identifies the full-future predictive distribution as a minimal sufficient state with recursive and Markov structure \cite{crutchfield_young1989,shalizi_crutchfield2001compmech}. Predictive inference and predictive rate-distortion study the representation and compression of future information \cite{still2014predictive,marzen_crutchfield2016prd}, while recent work considers operational predictive-state communication and the failure of finite predictive summaries to remain closed under updating \cite{ercetin_chraiti2026psc,wang2026fiber}. These results expose a distinction between sufficiency for a fixed prediction target and sufficiency for recursive sequential prediction.

Taken together, these lines of work provide the main ingredients for predictive sufficiency, loss-based fidelity, constrained receiver utilization, matched compression, and recursive predictive state. They are developed, however, around different source objects and operational questions, leaving their relationship implicit within prediction-oriented communication. The present work makes this relationship explicit by using a common target-conditioned predictive state to organize static fidelity, AI reception, compression, and sequential refinement. The resulting technical contributions are summarized next.

\subsection{Contributions and Paper Organization}

The central contribution is to make fidelity-object selection an explicit information-theoretic step and, for exact predictive preservation, to resolve that step through the target-conditioned predictive state. Once this object-selection principle is fixed, its consequences for fidelity, AI reception, compression, and sequential prediction follow. This yields three contributions.

\begin{enumerate}
\item \textbf{Predictive-fidelity object and geometry.}
We take the source-induced predictive state $S_X=P_{F\mid X,C}$ as the fidelity object, thereby grouping source realizations by the future predictions they support rather than by source similarity. We prove that $S_X$ is sufficient and minimal for exact predictive preservation. We then use the terminal prediction loss to define distortion directly on this predictive quotient, so that communication fidelity reflects predictive performance rather than source reconstruction error. Under logarithmic loss, this distortion is exactly the predictive information discarded by communication.

\item \textbf{AI-receiver loss decomposition.} Using the source- and receiver-side predictive states as common Bayes references, we derive an additive decomposition of end-to-end AI-receiver loss into three distinct mechanisms: predictive value lost in communication, value available at the receiver but unattainable by the predictor family, and family capability not realized by the deployed model. This decomposition reveals that identical end-to-end errors can arise from fundamentally different bottlenecks, each with its own zero-loss condition and system-level remedy.

\item \textbf{Compression sufficiency and sequential refinement.} We examine how far the predictive-state reduction continues to hold beyond fidelity definition. Under finite-alphabet memoryless assumptions and matched predictive-state distortion, access to raw-source distinctions within a predictive class provides no rate-distortion advantage over coding the state directly. For sequential prediction, a state minimal for a fixed horizon may fail to update recursively, while taking the entire future as the target yields a minimal recursively closed state with a Markov representation.
\end{enumerate}

Together, these contributions establish predictive-state preservation as a common abstraction for communication fidelity, AI-receiver loss attribution, matched compression, and sequential prediction. Figure~\ref{fig:framework} summarizes the resulting fixed-target framework and its sequential extension.

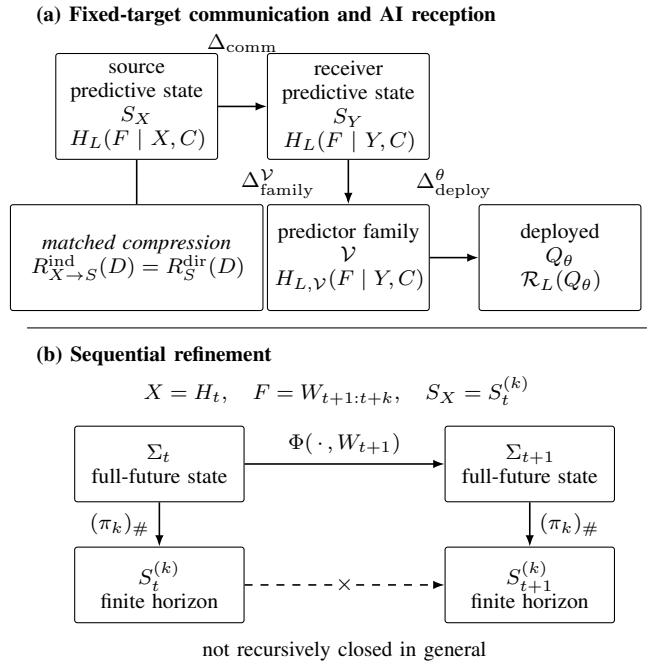
\begin{figure}[!b]
\centering
\ifCLASSOPTIONtwocolumn
\def\frameworkxunit{1cm}
\def\frameworkyunit{1cm}
\def\frameworkstaticwidth{20mm}
\def\frameworkdynamicwidth{21mm}
\def\frameworkcalloutwidth{32mm}
\def\frameworkstaticheight{14mm}
\else
\def\frameworkxunit{1.25cm}
\def\frameworkyunit{1.12cm}
\def\frameworkstaticwidth{27mm}
\def\frameworkdynamicwidth{26mm}
\def\frameworkcalloutwidth{35mm}
\def\frameworkstaticheight{15mm}
\fi
\begin{tikzpicture}[
x=\frameworkxunit,
y=\frameworkyunit,
static/.style={draw, rounded corners=1pt, align=center, inner sep=2pt, minimum height=\frameworkstaticheight, text width=\frameworkstaticwidth, font=\footnotesize},
dynamic/.style={draw, rounded corners=1pt, align=center, inner sep=2pt, minimum height=10mm, text width=\frameworkdynamicwidth, font=\footnotesize},
callout/.style={draw, rounded corners=1pt, align=center, inner sep=2pt, minimum height=\frameworkstaticheight, text width=\frameworkcalloutwidth, font=\footnotesize},
actual/.style={-{Latex[length=1.5mm]}, semithick},
blocked/.style={dashed, -{Latex[length=1.5mm]}, semithick}
]
\node[anchor=west,font=\footnotesize\bfseries] at (-4.20,0.38) {(a) Fixed-target communication and AI reception};

\node[static] (sx) at (-2.75,-0.82) {source predictive state\\$S_X$\\$H_L(F\mid X,C)$};
\node[static] (sy) at (0.05,-0.82) {receiver predictive state\\$S_Y$\\$H_L(F\mid Y,C)$};
\node[callout] (rd) at (-2.75,-2.82) {\textit{matched compression}\\\mbox{$R_{X\to S}^{\rm ind}(D)=R_S^{\rm dir}(D)$}};
\node[static] (qv) at (0.05,-2.82) {predictor family\\$\mathcal V$\\$H_{L,\mathcal V}(F\mid Y,C)$};
\node[static] (qt) at (2.85,-2.82) {deployed\\$Q_\theta$\\$\mathcal R_L(Q_\theta)$};

\draw[actual] (sx) -- (sy);
\draw[actual] (sy) -- (qv);
\draw[actual] (qv) -- (qt);
\draw[semithick] (sx) -- (rd);
\node[font=\footnotesize] at (-1.35,0.02) {$\Delta_{\rm comm}$};
\node[font=\footnotesize] at (-0.88,-1.84) {$\Delta_{\rm family}^{\mathcal V}$};
\node[font=\footnotesize] at (1.45,-1.84) {$\Delta_{\rm deploy}^{\theta}$};

\draw[thin] (-4.20,-3.75) -- (4.00,-3.75);
\node[anchor=west,font=\footnotesize\bfseries] at (-4.20,-4.12) {(b) Sequential refinement};
\node[font=\footnotesize] at (-0.10,-4.58) {$X=H_t,\quad F=W_{t+1:t+k},\quad S_X=S_t^{(k)}$};

\node[dynamic] (sigma) at (-2.45,-5.55) {$\Sigma_t$\\full-future state};
\node[dynamic] (sigman) at (2.45,-5.55) {$\Sigma_{t+1}$\\full-future state};
\node[dynamic] (sk) at (-2.45,-7.15) {$S_t^{(k)}$\\finite horizon};
\node[dynamic] (skn) at (2.45,-7.15) {$S_{t+1}^{(k)}$\\finite horizon};

\draw[actual] (sigma) -- node[midway,anchor=south,font=\footnotesize] {$\Phi(\,\cdot\,,W_{t+1})$} (sigman);
\draw[actual] (sigma) -- node[midway,anchor=east,font=\footnotesize] {$(\pi_k)_{\#}$} (sk);
\draw[actual] (sigman) -- node[midway,anchor=west,font=\footnotesize] {$(\pi_k)_{\#}$} (skn);
\draw[blocked] (sk) -- (skn);
\node[fill=white,inner sep=1pt,font=\footnotesize] at (0,-7.15) {$\times$};
\node[font=\footnotesize] at (0,-8.02) {not recursively closed in general};
\end{tikzpicture}
\caption{Target-conditioned predictive-fidelity framework. (a) For a fixed prediction target, the source and receiver predictive states provide common references for communication, model-family, and deployment losses, while matched compression acts directly on the predictive quotient. (b) For sequential prediction, the finite-horizon state is a marginal of the full-future state. The latter updates recursively, whereas the former need not be recursively closed.}
\label{fig:framework}
\end{figure}

\FloatBarrier

Section II defines the predictive state for exact preservation and establishes its minimality. Section III develops proper-loss fidelity and its log-loss information interpretation. Section IV develops the AI-receiver loss decomposition. Section V establishes the matched rate-distortion reduction. Section VI develops the sequential refinement and recursive-closure results. Section VII concludes the paper.

\section{Predictive State as the Fidelity Object for Exact Predictive Preservation}

\subsection{Problem Formulation and Assumptions}

Let $(C,X,F)\sim P_{C,X,F}$, where $C\in\mathcal C$ is shared context, $X\in\mathcal X$ is the source observation, and $F\in\mathcal F$ is the specified prediction target, typically a future random element.

We study exact predictive preservation. Given the shared context $C$, the goal is to identify what information in $X$ must be retained to preserve the entire predictive distribution of $F$, before introducing a particular distortion measure or coding scheme.

To ensure that the conditional distributions used below admit measurable versions, we assume that the random elements take values in standard Borel spaces. For every standard Borel space $\mathcal E$, we equip $\mathcal P(\mathcal E)$ with its evaluation $\sigma$-field. We fix one measurable version of each required conditional distribution, and all identities between such kernels are understood almost surely.

For sequential prediction, $F$ may represent a future block. The corresponding time-indexed specialization is introduced in Section~VI.

\subsection{Predictive State}

The first step is to represent what the source observation tells us about the prediction target once the shared context is known. This motivates using the conditional predictive distribution as the state.

\begin{definition}[Predictive state]
\label{def:1}

For a fixed context value $c$, define the predictive-state map

\[
\tau_c:\mathcal X\rightarrow\mathcal P(\mathcal F),
\qquad
\tau_c(x)
\triangleq
P_{F\mid X,C}(\,\cdot\mid x,c),
\]

where $\mathcal P(\mathcal F)$ is the space of probability measures over the target space $\mathcal F$. The source-induced predictive state is

\[
S_X
\triangleq
\tau_C(X)
=
P_{F\mid X,C}(\,\cdot\mid X,C).
\]

More generally, for any information variable $U$ jointly distributed with $(F,C)$, write

\[
S_U
\triangleq
P_{F\mid U,C}(\,\cdot\mid U,C).
\]

The subscript identifies the information variable inducing the state; dependence on the shared context is suppressed. The predictive state is relative to the prediction target, the available context, and the underlying joint distribution. Changing any of these generally changes $S_X$.

\end{definition}

The construction deliberately preserves the entire predictive distribution. If the terminal requires only a particular predictive property or Bayes action, a coarser task-specific representation may suffice \cite{sevetlidis2026bayes}.

\subsection{Predictive Equivalence and the Quotient Source}

To understand what information the predictive state retains from the raw source, we characterize source realizations that induce the same state through predictive equivalence.

\begin{definition}[Predictive equivalence]
\label{def:3}

For a fixed context $c$, two source realizations are predictively equivalent when they induce the same predictive state:
\[
x_1\sim_c x_2
\quad\Longleftrightarrow\quad
\tau_c(x_1)=\tau_c(x_2).
\]
The equality is understood as equality of probability measures on $\mathcal F$. To avoid distinctions caused only by off-support versions of the conditional distribution, interpret this relation on a full-$P_{X\mid C=c}$-measure set $\mathcal X_c$ for $P_C$-almost every $c$. Define
\[
\mathcal S_c
\triangleq
\tau_c(\mathcal X_c)
\subseteq
\mathcal P(\mathcal F)
\]
as the attainable predictive-state space.
\end{definition}

The map $\tau_c$ factors setwise through the predictive quotient:
\[
\mathcal X_c
\xrightarrow{\ \pi_c\ }
\mathcal X_c/{\sim_c}
\xrightarrow{\ \bar\tau_c\ }
\mathcal S_c,
\qquad
\bar\tau_c([x])=\tau_c(x).
\]
By construction, $\bar\tau_c$ is a bijection, so $\mathcal S_c$ provides a distribution-valued realization of the predictive quotient. This statement is set-theoretic; no measurable or topological quotient structure is required for the results below. Thus $\mathcal S_c$ is a concrete representation of the source distinctions that remain relevant to the specified prediction target.

\subsection{Sufficiency and Minimality}

Predictive equivalence so far gives only a partition of the source. To justify the predictive state as the fidelity object, we must establish that it preserves all information needed for the specified prediction and that every admissible representation with the same predictive capability must retain this state.

\begin{definition}[Admissible and predictively sufficient representation]
\label{def:4}
A possibly randomized representation $T$ is admissible if
\[
F-(X,C)-T.
\]
It is predictively sufficient if
\[
F\perp\!\!\!\perp X\mid(T,C).
\]
Thus, once $(T,C)$ is known, the raw observation $X$ provides no further information about the specified prediction target.
\end{definition}

The predictive state satisfies both requirements.

\begin{theorem}[Predictive sufficiency and minimality]
\label{thm:1}
Let
\[
S_X=P_{F\mid X,C}(\,\cdot\mid X,C).
\]
Then:
\begin{enumerate}
\item \textbf{Predictive sufficiency.} The state $S_X$ is an admissible function of $(X,C)$ and satisfies
\[
F\perp\!\!\!\perp X\mid(S_X,C).
\]

\item \textbf{Minimality.} For every admissible and predictively sufficient representation $T$, there exists a measurable map
\[
g:\mathcal T\times\mathcal C\rightarrow\mathcal P(\mathcal F)
\]
such that
\[
S_X=g(T,C)
\qquad\text{almost surely}.
\]
A valid version is
\[
g(t,c)=P_{F\mid T,C}(\,\cdot\mid t,c).
\]
Equivalently,
\[
\sigma(S_X,C)
\subseteq
\sigma(T,C)
\qquad (\bmod\ P).
\]
\end{enumerate}
\end{theorem}

\begin{proof}
Because $\mathcal F$ is standard Borel, let $\mathcal A_0$ be a countable determining class. For each $A\in\mathcal A_0$,
\[
P(F\in A\mid X,C)=S_X(A).
\]
The tower property gives
\[
P(F\in A\mid S_X,C)
=
\mathbb E[S_X(A)\mid S_X,C]
=
S_X(A).
\]
Since $S_X$ is determined by $(X,C)$,
\[
P(F\in A\mid X,S_X,C)
=S_X(A).
\]
The determining property of $\mathcal A_0$ therefore yields
\[
F\perp\!\!\!\perp X\mid(S_X,C).
\]

Now let $T$ be admissible and predictively sufficient. Admissibility and sufficiency imply, almost surely,
\[
P_{F\mid X,C}
=
P_{F\mid X,T,C}
=
P_{F\mid T,C}.
\]
Hence, with
\[
g(t,c)=P_{F\mid T,C}(\,\cdot\mid t,c),
\]
we obtain
\[
S_X=g(T,C),
\]
which proves the factorization and the corresponding inclusion of generated $\sigma$-fields.
\end{proof}

\begin{remark}[Meaning of minimality]
Minimality here refers to information content rather than representation size: every admissible predictively sufficient representation must contain enough information to recover $S_X$, although it may retain additional source information. Any other admissible predictively sufficient representation that is minimal in the same factor order is therefore equivalent to $S_X$ up to almost-sure measurable relabeling.
\end{remark}

\subsection{Reconstruction Boundary and Receiver State}

The predictive formulation should recover classical source identity when the prediction target itself is the source. Choosing $F=X$ gives

\[
S_X=P_{X\mid X,C}=\delta_X,
\qquad
x_1\sim_c x_2
\Longleftrightarrow
x_1=x_2.
\]

Thus the predictive quotient reduces to the original source space.

For a genuine future target, the map
\[
x\longmapsto P_{F\mid X,C}(\,\cdot\mid x,c)
\]
may instead be noninjective, allowing source realizations with identical predictive consequences to be merged. Thus the framework recovers exact source reconstruction when prediction requires source identity, while permitting a coarser fidelity object when it does not. This correspondence is with exact source reconstruction. Arbitrary lossy reconstruction metrics remain separately specified fidelity criteria.

Having established the source-side object and its reconstruction boundary, we now introduce the corresponding receiver-side state. Let $Y$ denote receiver-side information generated from $(X,C)$, so that
\[
F-(X,C)-Y.
\]
For every measurable $A\subseteq\mathcal F$,
\[
\begin{aligned}
S_Y(A)
&=
P(F\in A\mid Y,C)\\
&=
\mathbb E\!\left[P(F\in A\mid X,C)\mid Y,C\right]\\
&=
\mathbb E[S_X(A)\mid Y,C].
\end{aligned}
\]
Thus communication replaces the source-side predictive state by the conditional barycenter of the states still compatible with the received information. The mapping
\[
S_X\longrightarrow S_Y
\]
therefore provides the static predictive-fidelity interface between source and receiver.

\section{Predictive Fidelity and Information Accounting}

Section~II identifies the source and receiver predictive states and the static interface $S_X\to S_Y$. We now ask how predictive loss along this interface should be measured. We derive distortion from the terminal prediction loss and then show that, under logarithmic loss, the resulting fidelity admits an exact information-theoretic interpretation.

\subsection{Loss-Induced Fidelity on the Predictive Quotient}

We use the excess prediction risk induced by a proper probabilistic prediction loss as a fidelity measure between predictive states. Here proper means that reporting the true predictive distribution minimizes expected loss.

Let $\mathcal Q\subseteq\mathcal P(\mathcal F)$ be a class of admissible predictive reports containing the predictive states considered below. Equip $\mathcal Q$ with the trace $\sigma$-field inherited from $\mathcal P(\mathcal F)$, and let
\[
L:\mathcal Q\times\mathcal F\to\overline{\mathbb R}
\]
be a measurable terminal loss, where $L(Q,f)$ is the probabilistic prediction loss incurred by reporting $Q\in\mathcal Q$ when the realized target is $f$. For truth $P\in\mathcal Q$ and report $Q\in\mathcal Q$, define
\[
R_L(P,Q)
\triangleq
\mathbb E_{F\sim P}[L(Q,F)].
\]
We call $L$ proper on $\mathcal Q$ if truthful reporting satisfies
\[
R_L(P,P)\leq R_L(P,Q)
\]
for all $P,Q\in\mathcal Q$. It is strictly proper if equality can hold only when $Q=P$. Its Bayes risk is
\[
H_L(P)\triangleq R_L(P,P),
\]
Whenever the difference is well defined, define its loss-induced divergence, formally the regret divergence, by
\[
D_L(P\|Q)
\triangleq
R_L(P,Q)-R_L(P,P).
\]
Thus $D_L(P\|Q)$ is exactly the excess terminal prediction risk incurred by reporting $Q$ when the true predictive distribution is $P$. Under properness it is nonnegative, and under strict properness it vanishes only when $P=Q$. It is directed and need not be a metric.

For any information variable $U$ such that
\[
S_U=P_{F\mid U,C}(\,\cdot\mid U,C)\in\mathcal Q
\]
almost surely, define the conditional Bayes risk
\[
H_L(F\mid U,C)
\triangleq
\mathbb E[H_L(S_U)].
\]
This is the Bayes-optimal expected terminal risk when $(U,C)$ is available.

When $S_Y$ is admissible, it is Bayes-optimal among predictive reports based on $(Y,C)$. The distortion induced by communication should therefore equal the Bayes-risk increase from source-side to receiver-side information. The following identity makes this exact.

\begin{theorem}[Proper-loss communication identity]
\label{thm:2}
Let $Y$ satisfy $F-(X,C)-Y$. Let $L$ be proper on $\mathcal Q$, assume that $S_X,S_Y\in\mathcal Q$ almost surely, and suppose that all displayed expectations and risk differences are well defined. Then
\[
\boxed{
\mathbb E[D_L(S_X\|S_Y)]
=
H_L(F\mid Y,C)-H_L(F\mid X,C)
\geq 0.
}
\]
If $L$ is strictly proper, the common value is zero if and only if
\[
S_X=S_Y
\qquad\text{almost surely},
\]
or equivalently,
\[
F\perp\!\!\!\perp X\mid(Y,C).
\]
\end{theorem}

\begin{proof}
Conditioned on $(X,Y,C)$, the target distribution is $S_X$. Therefore
\[
\mathbb E[D_L(S_X\|S_Y)]
=
\mathbb E[L(S_Y,F)]-\mathbb E[L(S_X,F)].
\]
Because $S_X$ and $S_Y$ are the Bayes reports associated with $(X,C)$ and $(Y,C)$, respectively,
\[
\begin{aligned}
\mathbb E[L(S_X,F)]&=H_L(F\mid X,C),\\
\mathbb E[L(S_Y,F)]&=H_L(F\mid Y,C),
\end{aligned}
\]
which proves the identity and its nonnegativity.

Under strict properness, $D_L(S_X\|S_Y)=0$ exactly when $S_X=S_Y$. Since the divergence is nonnegative, its expectation vanishes exactly when this equality holds almost surely. Under the Markov relation, this is equivalent to
\[
F\perp\!\!\!\perp X\mid(Y,C).
\]
\end{proof}

Theorem~\ref{thm:2} completes the transition from fidelity-object selection to fidelity measurement. The prediction target determines which source distinctions remain relevant, while the terminal loss determines the cost of confusing the resulting predictive states. Different proper losses may therefore induce different fidelities on the same predictive quotient, with relative entropy arising specifically under logarithmic loss.

\subsection{Log-Loss Information Accounting}

For a general proper loss, predictive distortion is measured as excess terminal prediction risk. Logarithmic loss is distinguished by the fact that its induced divergence is relative entropy. We now show that this same predictive-fidelity quantity admits an exact Shannon-theoretic representation at the communication level.

Assume that the report class $\mathcal Q$ is dominated by a fixed reference measure $\mu$ on $\mathcal F$, and write $q=dQ/d\mu$ for $Q\in\mathcal Q$. For $P,Q\in\mathcal Q$, with $p=dP/d\mu$, logarithmic loss
\[
L_{\log}(Q,f)=-\log q(f)
\]
induces
\[
D_{L_{\log}}(P\|Q)
=
\mathbb E_{F\sim P}\!\left[\log\frac{p(F)}{q(F)}\right]
=
D_{\mathrm{KL}}(P\|Q).
\]
The divergence takes the extended value $+\infty$ when $P$ is not absolutely continuous with respect to $Q$.

Thus, under logarithmic loss, predictive-state fidelity is measured by relative entropy. The next result shows that, after averaging over the communication system, this same quantity is exactly the conditional mutual information about the prediction target that is unavailable from the receiver-side information.

For receiver information $Y$, define the expected log-loss predictive distortion

\[
D_p(X\to Y\mid C)
\triangleq
\mathbb E\!\left[
D_{\mathrm{KL}}\!\left(
P_{F\mid X,C}
\middle\|
P_{F\mid Y,C}
\right)
\right].
\]

\begin{theorem}[Log-loss predictive-information identity]
\label{thm:3}

Let $Y$ satisfy
\[
F-(X,C)-Y.
\]

Whenever the conditional relative entropies are well-defined,

\[
\boxed{
D_p(X\to Y\mid C)
=
I(F;X\mid Y,C).
}
\]

If the mutual-information terms are finite, then

\[
\boxed{
I(F;X\mid C)
=
I(F;Y\mid C)
+
D_p(X\to Y\mid C).
}
\]

\end{theorem}

\begin{proof}

The conditional-relative-entropy representation gives

\[
I(F;X\mid Y,C)
=
\mathbb E\!\left[
D_{\mathrm{KL}}\!\left(
P_{F\mid X,Y,C}
\middle\|
P_{F\mid Y,C}
\right)
\right].
\]

The Markov relation replaces $P_{F\mid X,Y,C}$ by $P_{F\mid X,C}$, proving the first identity. For the second identity, the chain rule and $F-(X,C)-Y$ give

\[
\begin{aligned}
I(F;X\mid C)
&=I(F;X,Y\mid C)\\
&=I(F;Y\mid C)+I(F;X\mid Y,C)\\
&=I(F;Y\mid C)+D_p(X\to Y\mid C).
\end{aligned}
\]
This completes the proof.
\end{proof}

When the relevant conditional entropies are well-defined and finite, the same quantity can also be written as
\[
D_p(X\to Y\mid C)
=
H(F\mid Y,C)-H(F\mid X,C).
\]
More importantly, predictive sufficiency of $S_X$ gives
\[
I(F;S_X\mid C)=I(F;X\mid C).
\]
Thus quotienting the raw source to its predictive state removes no information about the selected target. Under logarithmic loss, the communication-induced distortion $D_p(X\to Y\mid C)$ measures precisely the target information unavailable after the transition
\[
S_X\longrightarrow S_Y.
\]
Accordingly,
\[
D_p(X\to Y\mid C)=0
\]
if and only if the receiver is predictively sufficient, equivalently $S_X=S_Y$ almost surely, or $F\perp\!\!\!\perp X\mid(Y,C)$.

\section{AI-Receiver Predictive Fidelity}
\label{sec:ai-receiver}

Sections~II--III characterize the predictive content delivered to the receiver through the transition from $S_X$ to $S_Y$. This Bayes-level description is not yet the performance of a deployed AI receiver: the receiver may be unable to realize the Bayes-optimal prediction within its model family, and the deployed model may further fall short of the family optimum. We now separate these additional sources of predictive loss.

\subsection{Constrained Predictor Family and Deployed Model}

For a proper prediction loss, the receiver-side state
\[
S_Y=P_{F\mid Y,C}(\cdot\mid Y,C)
\]
is the Bayes-optimal report based on the complete receiver information $(Y,C)$, provided that it belongs to the admissible report class $\mathcal Q$. We use this ideal receiver as the reference point for separating model-family limitation from deployment suboptimality.

Let $\mathcal V$ be a nonempty class of measurable predictive kernels
\[
Q:(y,c)\longmapsto Q(\cdot\mid y,c)\in\mathcal Q,
\]
where $\mathcal V$ may encode architectural, computational, interface, or decoding restrictions. Let $Q_\theta\in\mathcal V$ denote the deployed predictor. For $Q\in\mathcal V$, define its terminal risk by
\[
\mathcal R_L(Q)
\triangleq
\mathbb E\!\left[
L\!\left(Q(\cdot\mid Y,C),F\right)
\right],
\]
and define the optimal risk attainable within the family by
\[
H_{L,\mathcal V}(F\mid Y,C)
\triangleq
\inf_{Q\in\mathcal V}\mathcal R_L(Q).
\]
When the infimum is attained, let
\[
Q_{\mathcal V}^{\star}
\in
\operatorname*{arg\,min}_{Q\in\mathcal V}\mathcal R_L(Q)
\]
denote a family-optimal predictor.

\subsection{Communication, Family, and Deployment Gaps}

Relative to the source-informed Bayes predictor, end-to-end AI-receiver loss can arise at three successive stages: predictive value may be lost before reception, the received information may not be fully exploitable by the predictor family, and the deployed model may fail to realize the best performance available within that family. The following proposition shows that these three effects admit an exact additive decomposition.

\begin{proposition}[AI-receiver predictive-fidelity decomposition]
\label{prop:ai-receiver-decomposition}

Let
\[
F-(X,C)-Y,
\]
let $L$ be proper on $\mathcal Q$, and suppose that
\[
S_X,S_Y\in\mathcal Q
\]
almost surely and that
\[
Q(\cdot\mid Y,C)\in\mathcal Q
\]
almost surely for every $Q\in\mathcal V$. Assume that all displayed risks are finite. Then
\[
\begin{aligned}
H_L(F\mid X,C)
&\leq H_L(F\mid Y,C)\\
&\leq H_{L,\mathcal V}(F\mid Y,C)\\
&\leq \mathcal R_L(Q_\theta).
\end{aligned}
\]
Define
\[
\begin{aligned}
\Delta_{\rm comm}
&\triangleq H_L(F\mid Y,C)-H_L(F\mid X,C)
=\mathbb E D_L(S_X\|S_Y),\\
\Delta_{\rm family}^{\mathcal V}
&\triangleq H_{L,\mathcal V}(F\mid Y,C)-H_L(F\mid Y,C)\\
&=\inf_{Q\in\mathcal V}
\mathbb E D_L\!\left(S_Y\middle\|Q(\cdot\mid Y,C)\right),\\
\Delta_{\rm deploy}^{\theta}
&\triangleq \mathcal R_L(Q_\theta)-H_{L,\mathcal V}(F\mid Y,C).
\end{aligned}
\]
All three gaps are nonnegative, and
\[
\boxed{
\mathbb E D_L\!\left(
S_X\middle\|Q_\theta(\cdot\mid Y,C)
\right)
=
\Delta_{\rm comm}
+
\Delta_{\rm family}^{\mathcal V}
+
\Delta_{\rm deploy}^{\theta}.
}
\]
\end{proposition}

\begin{proof}
For any $Q\in\mathcal V$, conditioning on $(Y,C)$ gives
\[
\mathcal R_L(Q)
=
H_L(F\mid Y,C)
+
\mathbb E\!\left[
D_L\!\left(
S_Y
\middle\|
Q(\cdot\mid Y,C)
\right)
\right].
\]
Taking the infimum over $\mathcal V$ yields the family gap, while $Q_\theta\in\mathcal V$ gives the deployment gap.

By Theorem~\ref{thm:2},
\[
H_L(F\mid Y,C)-H_L(F\mid X,C)
=
\mathbb E D_L(S_X\|S_Y),
\]
which gives the communication gap. Moreover,
\[
\mathbb E D_L\!\left(S_X\middle\|Q_\theta(\cdot\mid Y,C)\right)
=
\mathcal R_L(Q_\theta)-H_L(F\mid X,C).
\]
Inserting the receiver Bayes risk and the family-optimal risk between these endpoint risks gives the stated decomposition.
\end{proof}

The three gaps therefore have distinct operational meanings: $\Delta_{\rm comm}$ measures predictive value lost before reception, $\Delta_{\rm family}^{\mathcal V}$ measures receiver-available value that the predictor family cannot exploit, and $\Delta_{\rm deploy}^{\theta}$ measures family capability not realized by the deployed model. Only the first term has a universal Shannon-theoretic representation under logarithmic loss.

\subsection{Losslessness and Log-Loss Specialization}

The decomposition separates not only where predictive loss occurs, but also what would be required to eliminate it. It also clarifies which part of end-to-end AI-receiver loss has a Shannon-information interpretation. We consider these two consequences in turn.

\begin{corollary}[End-to-end AI-receiver losslessness]
\label{cor:ai-receiver-losslessness}

Under Proposition~\ref{prop:ai-receiver-decomposition}, the deployed AI receiver attains the source-informed Bayes risk,
\[
\mathcal R_L(Q_\theta)=H_L(F\mid X,C),
\]
or, equivalently,
\[
\mathbb E D_L\!\left(
S_X\middle\|Q_\theta(\cdot\mid Y,C)
\right)=0,
\]
if and only if
\[
\Delta_{\rm comm}
=
\Delta_{\rm family}^{\mathcal V}
=
\Delta_{\rm deploy}^{\theta}
=
0.
\]

If $L$ is strictly proper, the communication condition is equivalent to
\[
S_X=S_Y
\qquad\text{almost surely},
\]
or, under $F-(X,C)-Y$,
\[
F\perp\!\!\!\perp X\mid(Y,C).
\]
If, in addition, the family infimum is attained, then
\[
\Delta_{\rm family}^{\mathcal V}=0
\quad\Longleftrightarrow\quad
Q_{\mathcal V}^{\star}(\cdot\mid Y,C)=S_Y
\quad\text{almost surely},
\]
and
\[
\Delta_{\rm deploy}^{\theta}=0
\quad\Longleftrightarrow\quad
Q_\theta
\in
\operatorname*{arg\,min}_{Q\in\mathcal V}\mathcal R_L(Q).
\]

Thus, under strict properness and attainment, AI-receiver losslessness requires all three of the following:
\[
\begin{array}{ll}
\text{(i)}&
\text{the received information is predictively sufficient;}\\
\text{(ii)}&
\text{the receiver-side Bayes state is realizable by }\mathcal V;\\
\text{(iii)}&
\text{the deployed predictor attains the family optimum.}
\end{array}
\]
\end{corollary}

\begin{proof}
By Proposition~\ref{prop:ai-receiver-decomposition},
\[
\mathcal R_L(Q_\theta)-H_L(F\mid X,C)
=
\Delta_{\rm comm}
+
\Delta_{\rm family}^{\mathcal V}
+
\Delta_{\rm deploy}^{\theta},
\]
where all three terms are nonnegative. Hence the endpoint risks are equal if and only if every gap vanishes.

Under strict properness, zero expected loss-induced divergence implies equality of the corresponding predictive reports almost surely. This gives the communication statement and, when a family minimizer exists, the realizability statement. The deployment condition follows directly from the definition of the family infimum.
\end{proof}

\begin{remark}[Bayes versus AI-receiver losslessness]
Zero communication gap guarantees losslessness only relative to the unconstrained Bayes receiver. End-to-end AI-receiver losslessness also requires zero family and deployment gaps. Here losslessness means zero excess risk relative to the source-informed Bayes predictor, not zero intrinsic prediction risk. If the family infimum is not attained, $\Delta_{\rm family}^{\mathcal V}=0$ means arbitrarily accurate approximation in expected excess risk rather than exact realizability by a single member of $\mathcal V$.
\end{remark}

\begin{corollary}[Log-loss decomposition]
\label{cor:ai-receiver-logloss}

Let $L=L_{\log}$, with a common logarithm base for loss, entropy, relative entropy, and mutual information, and write
\[
\mathcal R_{\log}(Q)
\triangleq
\mathcal R_{L_{\log}}(Q).
\]
Whenever the displayed quantities are finite, Theorem~\ref{thm:3} identifies the communication gap exactly as conditional mutual information:
\[
\Delta_{\rm comm}
=
I(F;X\mid Y,C),
\]
and consequently
\[
\boxed{
\begin{aligned}
\mathcal R_{\log}(Q_\theta)
={}&
H(F\mid X,C)
+
I(F;X\mid Y,C)\\
&+
\Delta_{\rm family}^{\mathcal V}
+
\Delta_{\rm deploy}^{\theta}.
\end{aligned}
}
\]
Equivalently,
\[
\boxed{
\begin{aligned}
\mathbb E\!\left[
D_{\mathrm{KL}}\!\left(
S_X
\middle\|
Q_\theta(\cdot\mid Y,C)
\right)
\right]
={}&
I(F;X\mid Y,C)\\
&+
\Delta_{\rm family}^{\mathcal V}
+
\Delta_{\rm deploy}^{\theta}.
\end{aligned}
}
\]

The receiver-side gaps are
\[
\begin{aligned}
\Delta_{\rm family}^{\mathcal V}
&=
\inf_{Q\in\mathcal V}
\mathbb E\!\left[
D_{\mathrm{KL}}\!\left(
S_Y
\middle\|
Q(\cdot\mid Y,C)
\right)
\right]\\
&=
\inf_{Q\in\mathcal V}\mathcal R_{\log}(Q)
-
H(F\mid Y,C),
\end{aligned}
\]
and
\[
\Delta_{\rm deploy}^{\theta}
=
\mathcal R_{\log}(Q_\theta)
-
\inf_{Q\in\mathcal V}\mathcal R_{\log}(Q).
\]
\end{corollary}

\begin{proof}
Under log loss, $H_{L_{\log}}=H$ and $D_{L_{\log}}=D_{\mathrm{KL}}$. By Theorem~\ref{thm:3},
\[
\Delta_{\rm comm}
=
H(F\mid Y,C)-H(F\mid X,C)
=
I(F;X\mid Y,C).
\]
Substituting this identity into the decomposition of Proposition~\ref{prop:ai-receiver-decomposition} yields both boxed formulas. The expressions for the receiver-side gaps follow directly from their definitions under $D_{L_{\log}}=D_{\mathrm{KL}}$.
\end{proof}

Thus logarithmic loss identifies the communication term with Shannon information while leaving two receiver-side excess-risk terms. Communication removes predictive information before it reaches the receiver, whereas the family and deployment terms quantify failures to exploit information already available there.

\begin{table*}
\caption{Three receivers with the same one-bit deployed log-loss excess risk, each limited by a different layer.}
\label{tab:three-layer-logloss}
\centering
\footnotesize
\setlength{\tabcolsep}{2pt}
\begin{tabular}{@{}>{\raggedright\arraybackslash}p{0.13\textwidth}>{\raggedright\arraybackslash}p{0.07\textwidth}>{\raggedright\arraybackslash}p{0.15\textwidth}>{\raggedright\arraybackslash}p{0.17\textwidth}>{\centering\arraybackslash}p{0.10\textwidth}>{\centering\arraybackslash}p{0.12\textwidth}>{\centering\arraybackslash}p{0.11\textwidth}@{}}
\toprule
Limiting layer & $Y$ & Predictor family & Deployed predictor & $\Delta_{\rm comm}$ & $\Delta_{\rm family}^{\mathcal V}$ & $\Delta_{\rm deploy}^{\theta}$\\
\midrule
Communication & $A$ & $\mathcal V_Y$ & $Q_{1/2}=S_Y$ & $1$ & $0$ & $0$\\
Model family & $(A,B)$ & $\mathcal V_A$ & $Q_{1/2}$, family-optimal & $0$ & $1$ & $0$\\
Deployment & $(A,B)$ & $\mathcal V_Y$ & $Q_{1/2}$, suboptimal & $0$ & $0$ & $1$\\
\bottomrule
\end{tabular}
\end{table*}

\subsection{Restricted Receiver Interfaces and Layerwise Diagnosis}

The family gap in Proposition~\ref{prop:ai-receiver-decomposition} is defined abstractly for an arbitrary predictor family. Its meaning becomes particularly transparent when the receiver family can access the received information only through a restricted interface statistic. This case shows how predictive value can reach the receiver yet become unavailable to the deployed prediction pipeline.

\begin{remark}[Restricted receiver interface]
\label{rem:restricted-statistic-receiver}

Let
\[
U=\phi(Y,C)
\]
be the receiver-interface statistic, and suppose that the predictor retains separate access to the shared context $C$. Let $\mathcal V_U$ be the class of all admissible kernels of the form
\[
Q(\cdot\mid y,c)=\widetilde Q(\cdot\mid\phi(y,c),c),
\]
as $\widetilde Q$ ranges over measurable predictive reports based on $(U,C)$. Whenever $S_U\in\mathcal Q$ almost surely,
\[
H_{L,\mathcal V_U}(F\mid Y,C)
=
H_L(F\mid U,C),
\]
and hence
\[
\Delta_{\rm family}^{\mathcal V_U}
=
H_L(F\mid U,C)-H_L(F\mid Y,C)
=
\mathbb E D_L(S_Y\|S_U).
\]
Under log loss, whenever finite,
\[
\Delta_{\rm family}^{\mathcal V_U}
=
I(F;Y\mid U,C).
\]
Thus the family gap reduces to the predictive value in $Y$ that is discarded by the receiver interface $U$. If the predictor cannot separately access $C$, the relevant Bayes state and risk are instead $P_{F\mid U}$ and $H_L(F\mid U)$.
\end{remark}

This specialization shows that communication loss and receiver-interface loss can remove the same amount of predictive value while occurring at different stages of the system. The next example makes this distinction explicit.

\begin{example}[Equal end-to-end loss from three distinct layers]
\label{ex:three-layer-logloss}

Let $C$ be deterministic,
\[
A,B\overset{\mathrm{i.i.d.}}{\sim}\operatorname{Bernoulli}(1/2),
\qquad
X=(A,B),
\qquad
F=B,
\]
and use base-two logarithmic loss. For any information variable $Z$, let $\mathcal V_Z$ denote the class of all predictive kernels measurable with respect to $Z$, and let $Q_{1/2}$ always report $\operatorname{Bernoulli}(1/2)$.

Table~\ref{tab:three-layer-logloss} compares three receivers with the same deployed log-loss excess risk of one bit. In all three cases, $H(F\mid X)=0$, but exactly one of the three gaps is nonzero.

Although the three systems have identical end-to-end performance, their remedies are different. The communication-limited system requires a more informative representation or link; the family-limited system requires a richer interface or predictor family; and the deployment-limited system requires improved training, optimization, or calibration.
\end{example}

The restricted-interface case also makes explicit why representations carrying the same Bayes content need not be equally usable by a constrained predictor family: usability depends on how that information is exposed to $\mathcal V$.

\section{Matched Compression on the Predictive Quotient}
\label{sec:rd-consequence}

The preceding sections identify the predictive state as the fidelity object and derive its matched distortion without committing to a coding architecture. We now ask an operational question: once fidelity depends only on the predictive state, can access to additional distinctions in the raw source improve compression?

We answer this question under a finite-alphabet memoryless model by comparing indirect coding from the raw source with direct coding of the predictive state.

\subsection{Matched Direct and Indirect Coding Problems}

To isolate the value of raw-source information beyond the predictive state, we compare two coding problems under exactly the same predictive-state distortion. The indirect encoder observes the raw source, whereas the direct encoder observes only its predictive state.

For this section only, let

\[
(F_i,X_i,C_i),
\qquad i=1,\ldots,n,
\]

be i.i.d. copies of a finite-alphabet triple $(F,X,C)$, with $C^n$ available noncausally at both terminals. Define the single-letter predictive state

\[
S=\eta(X,C)
\triangleq
P_{F\mid X,C}(\cdot\mid X,C),
\]

and let

\[
S_i=\eta(X_i,C_i)
\]

denote its i.i.d. copies. Let $\mathcal S$ denote the finite support of $S$.

Fix a finite reproduction alphabet $\widehat{\mathcal S}$ satisfying

\[
\mathcal S\cup\widehat{\mathcal S}\subseteq\mathcal Q,
\]

and use the matched single-letter distortion

\[
d_p(s,\widehat s)
\triangleq
D_L(s\|\widehat s),
\qquad
(s,\widehat s)\in
\mathcal S\times\widehat{\mathcal S},
\]

assumed finite on this domain. This is the predictive-state fidelity developed in Section~III.

For the indirect coding problem, pull it back to the raw source alphabet through $\eta$:
\[
\widetilde d(x,c,\widehat s)
\triangleq
d_p\!\left(\eta(x,c),\widehat s\right)
=
D_L\!\left(\eta(x,c)\middle\|\widehat s\right).
\]
Thus the indirect and direct coding problems below are evaluated under the same matched predictive-state criterion and differ only in the encoder observation. In particular, $\widetilde d$ is constant on predictive-equivalence classes in its source argument.

An $(n,M_n)$ indirect code is a pair

\[
\begin{aligned}
f_n^X &: \mathcal X^n\times\mathcal C^n\to[1:M_n],\\
g_n^X &: [1:M_n]\times\mathcal C^n\to\widehat{\mathcal S}^{\,n},
\end{aligned}
\]

producing

\[
\widehat S_X^n
=
g_n^X\!\left(f_n^X(X^n,C^n),C^n\right).
\]

An $(n,M_n)$ direct state code is a pair

\[
\begin{aligned}
f_n^S &: \mathcal S^n\times\mathcal C^n\to[1:M_n],\\
g_n^S &: [1:M_n]\times\mathcal C^n\to\widehat{\mathcal S}^{\,n},
\end{aligned}
\]

producing

\[
\widehat S_S^n
=
g_n^S\!\left(f_n^S(S^n,C^n),C^n\right).
\]

For either problem, define the average distortion

\[
D_n
\triangleq
\frac1n
\sum_{i=1}^{n}
\mathbb E[d_p(S_i,\widehat S_i)],
\]

where $\widehat S^n$ denotes the corresponding reconstruction. A rate $R$ is achievable at distortion $D$ if there exists a sequence of the corresponding codes such that

\[
\begin{aligned}
\limsup_{n\to\infty}
\frac1n\log M_n
&\le R,\\
\limsup_{n\to\infty}D_n
&\le D.
\end{aligned}
\]

Let

\[
R_{X\to S}^{\mathrm{ind}}(D)
\quad\text{and}\quad
R_S^{\mathrm{dir}}(D)
\]

denote the infima of achievable rates in the indirect and direct problems, respectively. The two problems differ only in the encoder observation: the indirect encoder sees $(X^n,C^n)$, whereas the direct encoder sees the induced quotient sequence $(S^n,C^n)$.

\subsection{Predictive Quotient Reduction}

The indirect encoder appears to have an advantage because its test channel may depend on raw-source distinctions that the predictive state has discarded. The key question is whether such within-class dependence can improve matched fidelity. The next lemma shows that it cannot.

\begin{lemma}[Fiber-averaging reduction to the predictive quotient]
\label{lem:predictive-quotient}

For every test channel $W_{\widehat S\mid X,C}$ with reproduction alphabet $\widehat{\mathcal S}$, there exists a state test channel $Q_{\widehat S\mid S,C}$ such that the two channels induce the same joint law of $(S,C,\widehat S)$. Consequently, they achieve the same expected matched distortion, and

\[
I_Q(S;\widehat S\mid C)
\le
I_W(X;\widehat S\mid C).
\]

More precisely,

\[
\begin{aligned}
&I_W(X;\widehat S\mid C)
-I_Q(S;\widehat S\mid C)\\
&\qquad=
I_W(X;\widehat S\mid S,C)
\ge 0.
\end{aligned}
\]

\end{lemma}

\begin{proof}

For every pair $(s,c)$ with

\[
P_{S,C}(s,c)>0,
\]

define

\[
Q(\widehat s\mid s,c)
\triangleq
\sum_{x:\,\eta(x,c)=s}
P_{X\mid S,C}(x\mid s,c)
W(\widehat s\mid x,c),
\]
and define $Q$ arbitrarily on $P_{S,C}$-null pairs. Then

\[
\begin{aligned}
&P_{S,C}(s,c)Q(\widehat s\mid s,c)\\
&\quad=
\sum_{x:\,\eta(x,c)=s}
P_{X,C}(x,c)
W(\widehat s\mid x,c),
\end{aligned}
\]

which is exactly the joint probability of $(S,C,\widehat S)=(s,c,\widehat s)$ under the original channel $W$. Hence the construction preserves $P_{S,C,\widehat S}$ and therefore the matched distortion.

Since $S=\eta(X,C)$ is a deterministic function of $(X,C)$, the chain rule under the original channel $W$ gives

\[
\begin{aligned}
I_W(X;\widehat S\mid C)
&=
I_W(S;\widehat S\mid C)\\
&\quad+
I_W(X;\widehat S\mid S,C).
\end{aligned}
\]

Preservation of $(S,C,\widehat S)$ gives

\[
I_Q(S;\widehat S\mid C)
=
I_W(S;\widehat S\mid C).
\]

Combining these identities proves the result.
\end{proof}

Thus any dependence of the reconstruction on distinctions within a predictive-equivalence class consumes rate but cannot improve the matched distortion.

\subsection{Rate-Distortion Preservation}

Lemma~\ref{lem:predictive-quotient} establishes the reduction at the single-letter test-channel level. We now combine it with classical conditional rate-distortion theory to show that the same reduction holds operationally.

\begin{theorem}[Rate-distortion preservation on the predictive quotient]
\label{thm:rd-reduction}

Under the finite-alphabet memoryless model of Section~V-A, for every distortion level $D\ge 0$,

\[
\boxed{
\begin{aligned}
R_{X\to S}^{\mathrm{ind}}(D)
&=
\inf_{\substack{
W_{\widehat S\mid X,C}:\\
\mathbb E_W[d_p(S,\widehat S)]\le D
}}
I_W(X;\widehat S\mid C)\\
&=
\inf_{\substack{
Q_{\widehat S\mid S,C}:\\
\mathbb E_Q[d_p(S,\widehat S)]\le D
}}
I_Q(S;\widehat S\mid C)\\
&=
R_S^{\mathrm{dir}}(D).
\end{aligned}
}
\]

Thus access to raw-source distinctions within a predictive-equivalence class does not improve the rate-distortion tradeoff under matched predictive-state fidelity.

\end{theorem}

\begin{proof}

Classical conditional rate-distortion theory with $C^n$ available noncausally at both terminals \cite{gray1972conditional} gives

\[
R_{X\to S}^{\mathrm{ind}}(D)
=
\inf_{\substack{
W_{\widehat S\mid X,C}:\\
\mathbb E_W[d_p(S,\widehat S)]\le D
}}
I_W(X;\widehat S\mid C),
\]

with the source-space pullback $\widetilde d$ defined in Section~V-A, and

\[
R_S^{\mathrm{dir}}(D)
=
\inf_{\substack{
Q_{\widehat S\mid S,C}:\\
\mathbb E_Q[d_p(S,\widehat S)]\le D
}}
I_Q(S;\widehat S\mid C).
\]

For the first direction, every feasible raw-source channel $W_{\widehat S\mid X,C}$ is reduced by Lemma~\ref{lem:predictive-quotient} to a state channel $Q_{\widehat S\mid S,C}$ with the same joint law of $(S,C,\widehat S)$, and hence the same distortion, such that

\[
I_Q(S;\widehat S\mid C)
\le
I_W(X;\widehat S\mid C).
\]

Taking infima gives

\[
R_S^{\mathrm{dir}}(D)
\le
R_{X\to S}^{\mathrm{ind}}(D).
\]

For the reverse direction, given any feasible state channel $Q_{\widehat S\mid S,C}$, define a raw-source channel by

\[
W(\widehat s\mid x,c)
\triangleq
Q\!\left(\widehat s\mid\eta(x,c),c\right).
\]

This construction preserves the joint law of $(S,C,\widehat S)$ and therefore the expected distortion. Moreover,

\[
X-(S,C)-\widehat S
\]

under $W$, and $S$ is a deterministic function of $(X,C)$. Hence

\[
\begin{aligned}
I_W(X;\widehat S\mid C)
&=
I_W(S;\widehat S\mid C)
+I_W(X;\widehat S\mid S,C)\\
&=
I_W(S;\widehat S\mid C)\\
&=
I_Q(S;\widehat S\mid C).
\end{aligned}
\]

Thus every feasible state channel induces a feasible raw-source channel at the same rate and distortion, yielding

\[
R_{X\to S}^{\mathrm{ind}}(D)
\le
R_S^{\mathrm{dir}}(D).
\]

Combining the two inequalities proves the result.

\end{proof}

\begin{remark}[Matched-fidelity boundary]

The rate-distortion equivalence holds because the source-space distortion factors through the predictive state:

\[
\widetilde d(x,c,\widehat s)
=
d_p\!\left(\eta(x,c),\widehat s\right).
\]

Thus all source realizations within a predictive class are indistinguishable to the fidelity criterion. If a different distortion distinguishes such realizations, its rate-distortion function need not be preserved by the predictive quotient.

\end{remark}

\subsection{Log-Loss Closure and Information-Bottleneck Form}

Theorem~\ref{thm:rd-reduction} establishes the matched quotient reduction for a finite reproduction alphabet. Under logarithmic loss, allowing the full predictive simplex closes this compression result with the information-accounting identity of Section~III. Because the resulting setting falls outside the finite-reproduction assumptions of Theorem~\ref{thm:rd-reduction}, we state it as a separate finite-source extension.

Assume that $(F,X,C)$ is finite, let the reproduction alphabet be the full probability simplex $\mathcal P(\mathcal F)$, and use the extended-valued distortion

\[
d_{\log}(s,\widehat s)
\triangleq
D_{\mathrm{KL}}(s\|\widehat s).
\]

Restrict attention to codes with finite expected distortion, and denote the resulting indirect and direct operational rate-distortion functions by

\[
R_{X\to S}^{\mathrm{ind},\log}(D)
\quad\text{and}\quad
R_S^{\mathrm{dir},\log}(D).
\]

For any finite auxiliary $M$ satisfying $F-(X,C)-M$, consider reconstruction maps $\widehat s:\mathcal M\times\mathcal C\to\mathcal P(\mathcal F)$. The Bayes-optimal choice under log loss is the posterior predictive law \cite{courtade_weissman2014logloss}:

\[
\widehat S^{\star}(M,C)
=
P_{F\mid M,C}.
\]

By Theorem~\ref{thm:3}, applied with receiver information $Y=M$,

\[
\begin{aligned}
&\inf_{\widehat s}
\mathbb E\!\left[
D_{\mathrm{KL}}\!\left(
P_{F\mid X,C}
\middle\|
\widehat s(M,C)
\right)
\right]\\
&\quad=
\mathbb E\!\left[
D_{\mathrm{KL}}\!\left(
P_{F\mid X,C}
\middle\|
P_{F\mid M,C}
\right)
\right]\\
&\quad=
I(F;X\mid M,C).
\end{aligned}
\]

Thus the distortion bound becomes a constraint on predictive information. The fiber-averaging construction of Lemma~\ref{lem:predictive-quotient} applies to $M$ as well: any raw-source channel $P_{M\mid X,C}$ can be averaged within the fibers of $S=\eta(X,C)$ to preserve the joint law of $(S,C,M)$ while not increasing the coding rate. Together with the encoder Markov relation, predictive sufficiency gives $F\perp M\mid(S,C)$; preserving $(S,C,M)$ therefore preserves the posterior $P_{F\mid M,C}$ and the Bayes log-loss distortion. Conversely, every state channel can be pulled back through $\eta$ at the same rate and distortion.

Combining classical conditional log-loss source coding \cite{courtade_weissman2014logloss} with this quotient reduction yields
\[
\boxed{
\begin{aligned}
R_{X\to S}^{\mathrm{ind},\log}(D)
&=
R_S^{\mathrm{dir},\log}(D)\\
&=\inf_{\substack{
P_{M\mid S,C}:\\
|\mathcal M|<\infty,\\
I(F;X\mid M,C)\le D
}}
I(S;M\mid C).
\end{aligned}
}
\]

The relevance constraint can itself be expressed entirely on the predictive quotient. Indeed, $S$ is a deterministic function of $(X,C)$, predictive sufficiency gives $F\perp X\mid(S,C)$, and the state channel gives $M-(S,C)-X$. Hence $F\perp X\mid(S,M,C)$ and

\[
\begin{aligned}
I(F;X\mid M,C)
&=
I(F;X,S\mid M,C)\\
&=
I(F;S\mid M,C)
+I(F;X\mid S,M,C)\\
&=
I(F;S\mid M,C).
\end{aligned}
\]

Equivalently, the common rate-distortion function can be written as

\[
\inf_{\substack{
P_{M\mid S,C}:\\
|\mathcal M|<\infty,\\
I(F;S\mid M,C)\le D
}}
I(S;M\mid C),
\]

which is the familiar conditional information-bottleneck/log-loss form \cite{tishby1999ib,courtade_weissman2014logloss} entirely on the predictive quotient. The result remains a static finite-alphabet memoryless characterization; Section~VI turns to sequential predictive states at the state-representation level rather than through a dynamic coding theorem.

\section{Sequential Refinement and Recursive Closure}

Sections~II--V treat prediction relative to a fixed target. In sequential prediction, however, the target shifts as new observations arrive, so a state sufficient for the current target need not support recursive updating. We first show this limitation for finite-horizon states and then identify the full-future state as the minimal recursively closed refinement. The discussion concerns state representation rather than dynamic coding.

\subsection{Finite-Horizon Predictive State}

Let $\{W_t\}_{t\geq1}$ be a stochastic process on a standard Borel alphabet $\mathcal W$. Define

\[
H_t=W_{1:t},
\qquad
\mathcal H_t=\sigma(W_{1:t}),
\]

with $H_0$ the empty history and $\mathcal H_0$ the trivial $\sigma$-field. For a fixed horizon $k\geq1$, let

\[
F_t^{(k)}
\triangleq
W_{t+1:t+k}.
\]

At each time $t$, the static construction of Section~II applies with

\[
X=H_t,
\qquad
F=F_t^{(k)}.
\]

\begin{definition}[Finite-horizon predictive-state process]
\label{def:7}

The $k$-step predictive state at time $t$ is

\[
\begin{aligned}
S_t^{(k)}
&\triangleq
P\!\left(
W_{t+1:t+k}\in\cdot
\middle|
\mathcal H_t
\right)\\
&=
P_{W_{t+1:t+k}\mid H_t}
(\cdot\mid H_t).
\end{aligned}
\]

\end{definition}

By Theorem~\ref{thm:1}, $S_t^{(k)}$ is sufficient and minimal, in the measurable-factor order, for exact preservation of the current $k$-step predictive law. Sequential use requires a stronger property: whether there exists a measurable update map such that
\[
S_{t+1}^{(k)}
=
\Gamma_{k,t}\!\left(S_t^{(k)},W_{t+1}\right)
\qquad\text{almost surely}.
\]
The next subsection shows that fixed-target minimality alone does not guarantee such recursive closure.

\subsection{Finite-Horizon Sufficiency Does Not Imply Recursive Closure}

The obstruction comes from the moving target window: information outside the current horizon may be irrelevant now but become relevant after one shift. The following construction shows that this prevents a generally valid finite-horizon update.

\begin{proposition}[A counterexample to recursive closure]
\label{prop:2}
For every $k\geq1$, there exists a binary stochastic process and a time $t$ such that
\[
S_{t+1}^{(k)}
\]
is not measurable with respect to
\[
\sigma\!\left(S_t^{(k)},W_{t+1}\right).
\]
Consequently, there is no measurable map
\[
\Gamma_{k,t}:
\mathcal P(\mathcal W^k)\times\mathcal W
\longrightarrow
\mathcal P(\mathcal W^k)
\]
satisfying
\[
S_{t+1}^{(k)}
=
\Gamma_{k,t}\!\left(S_t^{(k)},W_{t+1}\right)
\qquad\text{almost surely}.
\]
\end{proposition}

\begin{proof}
Let
\[
A,B_1,\ldots,B_k
\overset{\mathrm{ind}}{\sim}
\operatorname{Bernoulli}(1/2),
\]
and define
\[
W_1=A,
\qquad
W_{j+1}=B_j,
\quad j=1,\ldots,k,
\qquad
W_{k+2}=A,
\]
extending the process arbitrarily thereafter.

At $t=1$, the current target $(B_1,\ldots,B_k)$ is independent of the revealed history $H_1=A$. Hence
\[
S_1^{(k)}
=
\operatorname{Unif}(\{0,1\}^k)
\qquad\text{almost surely}.
\]
After observing $W_2=B_1$, the shifted target becomes $(B_2,\ldots,B_k,A)$, so
\[
S_2^{(k)}
=
\operatorname{Bernoulli}(1/2)^{\otimes(k-1)}
\otimes\delta_A,
\]
with the first factor absent for $k=1$.

For either $b\in\{0,1\}$, the events $\{A=0,B_1=b\}$ and $\{A=1,B_1=b\}$ have positive probability and produce the same pair
\[
\bigl(S_1^{(k)},W_2\bigr)
=
\bigl(\operatorname{Unif}(\{0,1\}^k),b\bigr),
\]
but different values of $S_2^{(k)}$. Thus $S_2^{(k)}$ is not measurable with respect to $\sigma(S_1^{(k)},W_2)$.
\end{proof}

The failure is caused by horizon shift: $W_{t+k+1}$ lies outside the target defining $S_t^{(k)}$ but enters the next target. Information that is irrelevant to the current finite horizon can therefore become relevant one step later. Particular processes may admit closed finite-horizon updates, but such closure does not follow from predictive sufficiency alone.

\subsection{Full-Future Refinement}

The counterexample identifies what the finite-horizon state can discard: predictive distinctions that lie beyond the current target but may matter after later shifts. To retain all distinctions that can become relevant at any finite horizon, we now take the entire future as the prediction target.

Let
\[
W_{t+1:\infty}
\triangleq
(W_{t+1},W_{t+2},\ldots)
\]
denote the entire future tail.

\begin{definition}[Full-future predictive state]
\label{def:8}
The full-future predictive state at time $t$ is
\[
\Sigma_t
\triangleq
P_{W_{t+1:\infty}\mid H_t}
(\cdot\mid H_t)
\in
\mathcal P(\mathcal W^{\mathbb N}).
\]
\end{definition}

For $k\geq1$, let $\pi_k$ project onto the first $k$ future coordinates. Then
\[
S_t^{(k)}
=
(\pi_k)_{\#}\Sigma_t
\qquad\text{almost surely},
\]
and hence
\[
\sigma(S_t^{(k)})
\subseteq
\sigma(\Sigma_t)
\qquad
(\operatorname{mod}P).
\]
Thus $\Sigma_t$ determines every finite-horizon predictive state, whereas $S_t^{(k)}$ retains only the marginal required by the current target.

\subsection{Full-Future Sufficiency and Recursive Closure}

We now ask whether the full-future refinement has the two properties needed for sequential prediction: predictive sufficiency for the entire future and recursive closure under new observations.

\begin{definition}[Recursive predictive representation]
\label{def:9}
A state process $T_t=\rho_t(H_t)$ taking values in a standard Borel state space $\mathcal T_t$ is a \textbf{recursive predictive representation} if, for every $t$, it satisfies both:
\begin{enumerate}
\item \textbf{full-future sufficiency:}
\[
W_{t+1:\infty}
\perp\!\!\!\perp
H_t
\mid T_t;
\]
\item \textbf{recursive closure:} there exists a measurable update map
\[
\Psi_t:
\mathcal T_t\times\mathcal W
\longrightarrow
\mathcal T_{t+1}
\]
such that
\[
T_{t+1}
=
\Psi_t(T_t,W_{t+1})
\qquad\text{almost surely}.
\]
\end{enumerate}
\end{definition}

By measurable disintegration on standard Borel spaces \cite[Thm.~1.25]{kallenberg2017random}, we may fix a jointly measurable kernel
\[
\Phi:
\mathcal P(\mathcal W^{\mathbb N})
\times
\mathcal W
\longrightarrow
\mathcal P(\mathcal W^{\mathbb N}),
\]
where $\Phi(\sigma,w)$ gives the conditional law of tail coordinates $2,3,\ldots$ under $\sigma$ given first coordinate $w$, with arbitrary values on $\sigma_1$-null observations; here $\sigma_1$ is the first-coordinate marginal of $\sigma$.

\begin{theorem}[Full-future predictive state: minimality and recursive closure]
\label{thm:recursive-state}
Under the preceding standard Borel assumptions, the full-future state $\Sigma_t$ satisfies the following properties.
\begin{enumerate}
\item \textbf{Full-future sufficiency and minimality.}
\[
W_{t+1:\infty}
\perp\!\!\!\perp
H_t
\mid
\Sigma_t.
\]
Moreover, if $T_t=\rho_t(H_t)$ is any history statistic satisfying
\[
W_{t+1:\infty}
\perp\!\!\!\perp
H_t
\mid
T_t,
\]
then there exists a measurable map $g_t$ such that
\[
\Sigma_t
=
g_t(T_t)
\qquad\text{almost surely}.
\]
Equivalently,
\[
\sigma(\Sigma_t)
\subseteq
\sigma(T_t)
\qquad
(\operatorname{mod}P).
\]
\item \textbf{Recursive update.}
\[
\Sigma_{t+1}
=
\Phi(\Sigma_t,W_{t+1})
\qquad\text{almost surely}.
\]
\item \textbf{Markov representation.}
The process $\{\Sigma_t\}$ admits the transition kernel
\[
K(\sigma,B)
\triangleq
\int_{\mathcal W}
\mathbf 1
\bigl\{
\Phi(\sigma,w)\in B
\bigr\}
\sigma_1(dw),
\]
for $B\in\mathcal B(\mathcal P(\mathcal W^{\mathbb N}))$.
\end{enumerate}
Consequently, $\{\Sigma_t\}$ is itself a recursive predictive representation and is a measurable factor of every other full-future sufficient representation, in particular of every recursive predictive representation.
\end{theorem}

\begin{proof}
Full-future sufficiency and minimality follow directly from Theorem~\ref{thm:1} by taking
\[
X=H_t,
\qquad
F=W_{t+1:\infty},
\]
with trivial shared context.

For recursive updating, conditioned on $H_t$, $\Sigma_t$ is the law of $(W_{t+1},W_{t+2},\ldots)$. After observing $W_{t+1}$, the disintegration kernel $\Phi$ gives the conditional law of the remaining tail. Hence, by uniqueness of regular conditional laws,
\[
\Sigma_{t+1}
=
\Phi(\Sigma_t,W_{t+1})
\qquad\text{almost surely}.
\]

Finally, conditioned on $\mathcal H_t$, the next observation has law $(\Sigma_t)_1$, and the next state is $\Phi(\Sigma_t,W_{t+1})$. Joint measurability of $\Phi$ makes $K$ a Markov kernel, and for every measurable state set $B$,
\[
\begin{aligned}
P(\Sigma_{t+1}\in B\mid\mathcal H_t)
&=
\int_{\mathcal W}
\mathbf 1
\bigl\{
\Phi(\Sigma_t,w)\in B
\bigr\}
(\Sigma_t)_1(dw)\\
&=
K(\Sigma_t,B).
\end{aligned}
\]
The right-hand side depends on the past only through $\Sigma_t$. Conditioning again on $\sigma(\Sigma_0,\ldots,\Sigma_t)\subseteq\mathcal H_t$ establishes the Markov representation.
\end{proof}

The two states therefore solve different problems. The finite-horizon state is minimal for exact preservation of the current $k$-step predictive law, whereas the full-future state is minimal for preserving the entire future law and is recursively closed under new observations. Fixed-target predictive sufficiency and recursive sequential sufficiency are therefore distinct requirements.

\section{Conclusion}

This paper starts from a question that precedes distortion design and
coding: when communication serves prediction, what should be preserved?
For a specified prediction target and shared context, we identify the
source-induced predictive state
$S_X=P_{F\mid X,C}$ as the fidelity object for exact predictive preservation. It
identifies source realizations by the predictive distribution they induce
and is minimal among admissible representations preserving that law.
This makes fidelity-object selection an explicit first step of
prediction-oriented communication, rather than treating the raw source
as the default object of preservation.

Once this principle for selecting the fidelity object is fixed, the rest of
the framework follows. Terminal
prediction loss induces predictive fidelity, with logarithmic loss
yielding an exact accounting of predictive information lost through
communication. Source- and receiver-side predictive states separate
communication, model-family, and deployment losses. Under matched
memoryless compression, raw-source distinctions beyond the predictive
quotient provide no rate-distortion gain. In sequential prediction,
fixed-horizon sufficiency need not imply recursive closure, motivating
the full-future predictive state for recursive evolution.

The central message is therefore a change in where communication design
begins. For prediction-oriented and AI-native systems, the first question
is not how accurately to reconstruct the source, nor even which
distortion to optimize, but which source distinctions matter to the
terminal objective in the first place. Distortion, compression,
transmission, and receiver design come after that choice. The shift from
source reconstruction to predictive-state preservation therefore begins
with a change in the first question: before asking how to preserve
information, first ask what information needs to be preserved.

\bibliographystyle{IEEEtran}
\bibliography{references}

\end{document}